\documentclass[12pt]{article}
\usepackage[T1]{fontenc}
\usepackage[margin=1in]{geometry} 
\usepackage[onehalfspacing]{setspace} 
\usepackage{graphicx}
\usepackage{microtype, lmodern}

\usepackage{natbib}
\usepackage[usenames, dvipsnames]{xcolor}
\usepackage[colorlinks=true]{hyperref}
\hypersetup{linkcolor=blue, citecolor=blue}

\usepackage{amsmath, amssymb, amsthm, thmtools, dsfont, bbm, enumerate}
\allowdisplaybreaks 

\usepackage{titlesec}
\titlespacing{\paragraph}{0pt}{2.25ex plus 1ex minus .2ex}{0.4em}

\theoremstyle{definition} 

\theoremstyle{plain} 
\newtheorem{assumption}{Assumption}

\newtheorem{theorem}{Theorem}

\begin{document}

\onehalfspacing

\title{Panel Data Estimation of Individual Demand in Markets with Many Consumers}
\author{Sarah Moon and Whitney K. Newey}
\date{\today}

\maketitle

\begin{abstract}
The purpose of this paper is to consider whether and how panel data can be used to estimate individual demand, as opposed to market-level demand, while accounting for simultaneity resulting from prices being determined in markets. We consider linear demand models and random coefficient demand models, together with linear supply models. We find that the bias of individual demand estimates obtained using familiar panel data methods, like differencing, disappears as the number of consumers in each market grows, as long as the time-varying, i.e. idiosyncratic, component of preferences is orthogonal to the unobserved, time-varying component of supply. This approximate control is assumed in many panel discrete choice models and is plausible in other models where idiosyncratic preferences represent random variation in preferences over time. Macroeconomic effects can be allowed for by including regressors characterizing time effects, such as trends and time period dummies, or fixed time effects.
\end{abstract}

\section{Introduction}
It has long been known that fixed effects estimation controls for time-invariant endogeneity in panel data for individuals. 
The purpose of this paper is to consider whether and how panel data can be used to estimate individual demand effects while controlling for endogeneity resulting from prices being determined in markets. 
We find that fixed effects estimation of price effects for individual demand is approximately correct when the number of consumers in a market is large and idiosyncratic (individual-specific, time-varying)  heterogeneity is independent of unobserved supply disturbances.
Such independence seems plausible in many settings where consumers in a market are distinct geographically or professionally from those who supply that market.

We obtain these results for a demand model similar to \cite{kennan1989simultaneous} where market demand is the sum of demands of individuals facing a common market price that is determined in equilibrium where market demand equals market supply.
We allow there to be many more consumers in a market than are present in the data. 
Approximate consistency holds because each individual's contribution to the equilibrium price becomes negligible as the number of individuals in the market grows large. 
This means prices are approximately exogenous from any single consumer's perspective, even though they are endogenous at the market level. 
We derive the rate at which bias in individual demand estimates shrinks as the number of individuals in the market increases.

We innovate in considering panel data for individual consumers where market prices vary over time.
The key insight is that time variation in the market prices, combined with idiosyncratic preferences that are independent of supply shocks, leads to fixed effects estimation being approximately correct.  

Independence between idiosyncratic preferences and supply shocks could be a concern if macroeconomic time effects were common to both.
Time effects may be controlled for by including observable time variables such as trends or time period dummies.
There are some limits to controlling for time effects, for example, in losing identification when there are no restrictions on how demand varies across time and individuals \citep*{hoderlein2012nonparametric, chernozhukov2013average}. 
As in many other settings, including time period dummies could provide a plausible intermediate approach that allows identification of important properties of individual demand from panel data.
We find that with two-way fixed effects (individual and time), estimation is approximately correct when individuals are sampled from different markets.

We first consider a linear model of individual demand with additive individual fixed effects and idiosyncratic disturbances.
In this well-understood setting, it is easy to see how independence of supply and idiosyncratic demand shocks provides identifying information while fixed individual effects allow for endogeneity of prices.
We then generalize these results to allow for fixed time period effects in individual demand and to allow for a random coefficient on price. 
In all settings we consider a linear average supply function for tractability. 
The results of this paper can be extended to general supply functions, a direction we are currently pursuing, but the extension requires slightly different proof techniques. The results can also be easily extended to allow for additional additive covariates in the demand and supply models, but we omit them in this paper for notational simplicity.

Throughout the paper we characterize the rate at which bias of the fixed effects estimator vanishes as the number of individuals in a market grows, rather than seeking strict identification. 
When market demand is the sum of individual demands, each individual's influence on the equilibrium price always results in a small but nonzero bias.  
We show that this bias shrinks as the number of consumers in a market increases and characterize the rate at which it shrinks, so that individual price effects are approximately identified from panel data.

This paper motivates the approach taken to estimating price effects using scanner data in \cite*{chernozhukov2019demand}.
That paper relied on individual-specific coefficient models to estimate average equivalent variation and deadweight loss from a panel of consumers. 
The results given in our paper suggest that this approach may be approximately correct when idiosyncratic preferences are independent of supply shocks.
This paper also motivates the approach of \cite*{dubois2020well} and \cite*{lanier2023estimating}. 
In their model the idiosyncratic preferences are i.i.d. Logit (Type I extreme value) disturbances.
If the assumption that the idiosyncratic preferences are independent of supply shocks is added, the results given in our paper suggest approximate correctness of their approach, although for a different model of demand.
We emphasize that these are conjectures at this time but initial results are promising.

This paper relates to several strands of literature. 
A large literature in industrial organization identifies and estimates market-level demand using instruments \citep*{berry1995automobile, petrin2002quantifying, berry2004differentiated, berry2014identification, berry2024nonparametric}.
We innovate in showing that for panel data on individuals price effects can be identified when idiosyncratic preferences are independent of unobserved supply shocks, without instruments. 
This result also differs from that of \cite*{mackay2025estimating}, who show that with covariance restrictions, market-level demand can be identified and estimated.
Here only idiosyncratic preferences are required to be independent of unobserved supply effects, while time-constant unobserved preferences may still be correlated with supply shocks. This is a different restriction from requiring market demand and supply unobservables to be uncorrelated. 

A similar comparison can be made with an older literature on identifying structural coefficients under certain covariance restrictions on the structural errors in the models \citep*{hausman1983identification, hausman1987efficient}. 
Panel data allows us to modify the zero covariance restriction to only apply to idiosyncratic preferences rather than to the whole structural error.

In the econometrics literature, \cite*{hoderlein2014nonparametric} use panel data to estimate market-level demand, where each individual unit is a market. \cite*{hausman2016individual} and \cite*{chernozhukov2025fisher} estimate functionals of individual demand using panel data, where the individual unit is a consumer. While these methods do not consider price endogeneity, we motivate these methods by showing price endogeneity can be small under certain conditions.

Section \ref{sec2} of the paper gives results for a linear panel demand model without a time effect and Section \ref{sec3} for a linear panel demand model with a time effect. In Section \ref{sec4} we derive results for a random coefficient demand model without a time effect and in Section \ref{sec5} we derive results for a random coefficient demand model with a time effect. Proofs of results are deferred to the Appendix.

\section{A Linear Demand Model for Panel Data} \label{sec2}
In this section, we consider a linear demand model and illustrate why, in panel data, when estimating individual demand, bias caused by price endogeneity might disappear with the number of consumers in the market if time differences of individual demand and aggregate supply shocks do not covary and preference variation is uncorrelated across individuals.
We suppose price is determined by linear market supply and individual demand equations. 
We give conditions for time differences in the market price to be nearly uncorrelated with idiosyncratic preference differences, so that a regression of quantity differences on market price differences would be approximately consistent. 

Consider a single market with $M$ consumers $i = 1, \dots M$. For a sample of size $N$ of those consumers, we observe quantity $q_{it}$ and price $p_t$ of a single good for each consumer $i$ over $T \geq 2$ time periods $t = 1, \dots, T$.
Suppose we can write demand for every consumer in the market as a linear function of price,
\begin{align*}
    q_{it} = \beta p_t + \alpha_i + \eta_{it},
\end{align*}
where $\alpha_i$ is unobserved time-invariant individual heterogeneity and $\eta_{it}$ is unobserved idiosyncratic individual heterogeneity.
Then average demand in each time period $t$, where the average is taken over all individuals in the market, is given by
\begin{align*}
    \bar q_t &= \beta p_t + \bar \alpha + \bar \eta_{t}, \qquad \text{for } \bar q_t \equiv \frac{1}{M} \sum_{i=1}^M q_{it}, \; \bar \alpha \equiv \frac{1}{M} \sum_{i=1}^M \alpha_i, \; \bar \eta_t \equiv \frac{1}{M} \sum_{i=1}^M \eta_{it}.
\end{align*}

Suppose, as in \cite{kennan1989simultaneous}, that we can write the inverse supply model as a linear model of average supply,
\begin{align*}
    \bar q_t &= \phi + \delta p_t + \varepsilon_t,
\end{align*}
for $\varepsilon_t$ an unobserved supply shock.
Solving for equilibrium price $p_t$ by setting demand equal to supply then gives
\begin{align*}
    p_t = \frac{1}{\beta-\delta} (\phi + \varepsilon_t - \bar{\alpha}- \bar{\eta}_{t} ).
\end{align*}

As usual for panel data we can difference these equations to eliminate the individual effect, obtaining 
\begin{align*}
   \Delta q_{it} = \beta \Delta p_t + \Delta\eta_{it}, \qquad \Delta p_t = \frac{1}{\beta-\delta} (\Delta\varepsilon_t - \Delta\bar\eta_{t} ),
\end{align*}
where $\Delta$ denotes the time difference. 
This is a model of linear supply and demand that possesses the classic price endogeneity problem. 
Due to the simultaneity of the system, price $p_t$ is correlated with individual unobserved preferences $\alpha_i$ and $\eta_{it}$. This dependence could cause naive estimates of the individual demand equation that do not exploit instrument-like conditions to be inconsistent. However, in what follows, we show that this inconsistency will disappear as the number of consumers in the market $M$ grows, as long as time differences in supply shocks are uncorrelated with time differences in individual idiosyncratic preference shocks $\eta_{it}$. In particular, the inconsistency will disappear at the rate $O(M^{-1})$.

The key condition for the inconsistency to vanish as $M$ grows is the following one. 

\begin{assumption}\label{ass1}
    $Cov(\Delta\varepsilon_t,\Delta\eta_{it}) = 0$ and
    $Cov(\Delta\eta_{it},\Delta\eta_{jt})=0$ for all $i\neq j$.
\end{assumption}

Assumption \ref{ass1} states that the time difference of the supply shock $\Delta\varepsilon_t$ is orthogonal to the time difference of the preference shock and that the preference shocks for different individuals are uncorrelated with each other. 
We will additionally assume for notational simplicity that $Var(\Delta \eta_{it})$ is common across $i$ and $t$ and $Var(\Delta p_t)$ is common across $t$.
Under this condition, 
\begin{align*}
    Cov(\Delta p_t,\Delta \eta_{it}) 
    &= \frac{1}{\beta-\delta}( Cov(\Delta\varepsilon_t,\Delta\eta_{it}) -
    Cov(\Delta\bar\eta_{t},\Delta\eta_{it})) =  \frac{-1}{\beta-\delta}\frac{Var(\Delta \eta_{it})}{M}.
\end{align*}
Then we immediately obtain the following result:

\begin{theorem}\label{thm:linear}
    Suppose Assumption \ref{ass1} holds. Then the slope estimate $\hat{\beta}$ from the regression of $\Delta q_{it}$ on  $\Delta p_t$ will have probability limit as $N,T \to \infty$ 
    \begin{align*}
    plim(\hat\beta)=\beta+\frac{Cov(\Delta p_t,\Delta \eta_{it}) }{Var(\Delta p_t)}
    &=\beta + \frac{1}{M} \frac{-1}{\beta-\delta} \frac{Var(\Delta \eta_{it})}{Var(\Delta p_t)}.
\end{align*}
\end{theorem}

In this result we assume that there is sufficient independent variation in $\Delta p_t$ across observations to make the usual formula for the probability limit of OLS valid.
Such variation could come from $T$ growing or from $T$ fixed and sampling consumers from independent markets with common $Var(\Delta p_t)$ and $M$.
From this formula we see that as the number of consumers $M$ grows $plim(\hat\beta)$ converges to the true $\beta$. 
In this result we think of taking $N,T \to \infty$ but keeping $M$ fixed as an asymptotic approximation of a large sample size. While we can also perform a finite-sample analysis of $\hat\beta$ and bound the size of the bias, we will not pursue this for the linear demand and supply models studied in this paper. 

This result demonstrates that under certain conditions, even when prices are endogenous due to simultaneity of the demand-supply system, standard panel data differencing can still recover individual demand slopes as markets become large. This follows because the source of price endogeneity when estimating individual demand is that equilibrium prices depend on individual preferences through average market demand. However, as the market grows large, the influence of any one individual's preferences on the market price becomes negligible, and we obtain the result above.

The usefulness of this result hinges on the plausibility of Assumption \ref{ass1}.
An important feature of Assumption \ref{ass1} is that it restricts only the time-varying component $\eta_{it}$ of the unobserved determinant $\alpha_i+\eta_{it}$ of preferences.
The fixed preference component $\alpha_i$ is allowed to covary with the supply shock in any way at all.
In many demand models $\eta_{it}$ is specified to be independent of all other variables, e.g. in the original Logit choice models of \cite{mcfadden1974conditional} and \cite*{berry1995automobile} where $\eta_{it}$ has a Type I Extreme Value distribution.
Such a specification is commonly used for other nonlinear panel data models as well, e.g. \cite{chamberlain1984panel}.
One way to motivate such a specification is to regard $\eta_{it}$ as a time-varying shock to preferences that represents individual taste for variety. 
Such time-varying individual preferences are generally allowed for in panel data demand models because they better fit demand data.
It seems plausible that these would be unrelated to supply shocks as long as the consumer is not employed by the firms supplying the product of interest.

One potential source of covariation in changes in preferences and supply is macroeconomic shocks that are common to supply and demand.
One way these can be allowed for in this model is through the use of observable regressors that depend on time.
For example, the above formula can allow for the presence of a constant in the panel differences regression.
The presence of a constant in the differences regression is equivalent to the presence of a time trend if several time periods are used in the regression or a time dummy in a two-period regression across different markets.
The nature of Assumption \ref{ass1} in ruling out covariance between time differences in preferences and supply shocks does indicate that it is important to check for time effects in panel demand models. As such, we now extend the model to allow for a fixed time effect, where the result will hold under different but related conditions.

\section{A Linear Demand Model with Time Effect} \label{sec3}

In this section, we extend the linear demand model of the previous section to allow for a fixed time effect. We show that bias caused by price endogeneity will disappear as the number of consumers in each market grows if individual demand and supply shocks are orthogonal across individuals and over time, and if preferences are uncorrelated across individuals and over time.

Suppose there is a good sold in many markets, each of which has $M$ consumers. We observe a sample of $N$ different consumers, $i = 1, \dots, N$, each from a different market $m$, over $T \geq 2$ time periods $t = 1,\dots,T$. In this sample we observe the quantity demanded $q_{it}$ and price $p_{it}$ of the good for each consumer $i$ in each time period $t$. Note that although we will assume in our model that prices do not vary across individuals in the same market, in the sample prices vary across individuals because each individual in the sample is drawn from a different market. This price variation is important for showing consistency for the linear demand model with a time effect. If all individuals come from the same market, in our model they face the same price and so prices would be collinear with the time effect.

Suppose that for each individual in every market we can write demand as a linear function of price,
\begin{align*}
    q_{it} = \beta p_{it} + \alpha_i + \phi_t + \eta_{it},
\end{align*}
where $\alpha_i$ is unobserved time-invariant individual heterogeneity, $\phi_t$ is fixed and unobserved individual-invariant time heterogeneity, and $\eta_{it}$ is unobserved idiosyncratic individual heterogeneity.

Let $\{i: i \in m\}$ denote the set of individuals $i$ in market $m$, which may include individuals not in the sample. Because price is common across all individuals in the same market, when analyzing a single market we will denote price as $p_{mt}$. The average quantity demanded in market $m$, where the average is taken over all individuals in market $m$, is given by
\begin{align*}
    \bar q_{mt} &= \beta p_{mt} + \bar \alpha_m + \phi_t + \bar \eta_{mt}, \qquad \text{for } \bar q_{mt} \equiv \frac{1}{M} \sum_{i \in m} q_{it}, \; \bar \alpha_m \equiv \frac{1}{M} \sum_{i \in m} \alpha_i, \; \bar \eta_{mt} \equiv \frac{1}{M} \sum_{i \in m} \eta_{it}.
\end{align*}

Suppose as before that in each market $m$ we can write the inverse supply model as a linear model of average supply,
\begin{align*}
    \bar q_{mt} &= \gamma + \delta p_{mt} + \varepsilon_{mt},
\end{align*}
where $\varepsilon_{mt}$ is an unobserved supply shock in market $m$.
Solving for equilibrium prices $p_{mt}$ in market $m$ by setting demand equal to supply, we obtain
\begin{align*}
    p_{mt} = \frac{1}{\beta-\delta} (\gamma + \varepsilon_{mt} - \bar{\alpha}_m - \phi_t - \bar{\eta}_{mt} ).
\end{align*}
Because we assume that each individual in the sample comes from a different market, in what follows we will use $\varepsilon_{it}$ instead of $\varepsilon_{mt}$ and $\bar\eta_{it}^m$ instead of $\bar\eta_{mt}$ to denote the observation for the individual $i$ in the sample coming from market $m$.

As usual for panel data with two-way fixed effects, for each individual $i$ in the sample we can subtract off sample unit-specific means and time-specific means to eliminate the unobserved fixed effects. In particular, 
\begin{align}
    \underbrace{q_{it} - \bar q_t - \bar q_i + \bar q}_{\displaystyle \tilde q_{it}} &= \beta \underbrace{\left(p_{it} - \bar p_t - \bar p_i + \bar p \right)}_{\displaystyle \tilde p_{it}} + \underbrace{\eta_{it} - \bar\eta_t - \bar\eta_i + \bar\eta}_{\displaystyle \tilde \eta_{it}} \label{eqn:est}, \\
    \underbrace{\left(p_{it} - \bar p_t - \bar p_i + \bar p \right)}_{\displaystyle \tilde p_{it}} &= \frac{1}{\beta-\delta} \big( \underbrace{\varepsilon_{it} - \bar \varepsilon_t - \bar \varepsilon_i + \bar \varepsilon}_{\displaystyle \tilde \varepsilon_{it}} - (\underbrace{\bar\eta_{it}^m - \bar \eta_t^m - \bar \eta_i^m + \bar \eta^m}_{\displaystyle \tilde \eta_{it}^m}) \big), \label{eqn:equil}
\end{align}
where variables with a time subscript $t$ are averaged over individuals $i$ in the sample, variables with an individual subscript $i$ are averaged over time periods $t$ in the sample, and variables without any subscript are averaged over both individuals $i$ and time periods $t$ in the sample. The exact definitions of all variables are available in the Appendix.

In general $\tilde p_{it}$ will not be orthogonal to $\tilde\eta_{it}$ due to price endogeneity caused by simultaneity of the system. There are two potential sources of endogeneity. The first is that $\tilde p_{it}$ depends on $\tilde \eta_{it}^m$, which is correlated with $\tilde \eta_{it}$, as can be seen in \eqref{eqn:equil}. The second, which we address with our key assumption, is that $\tilde p_{it}$ depends on $\tilde \varepsilon_{it}$, and supply shocks $\tilde \varepsilon_{it}$ may not be uncorrelated with idiosyncratic preferences $\tilde\eta_{it}$. Thus OLS on \eqref{eqn:est} will not be consistent. However, we show that this inconsistency will disappear as the number of individuals in each market $M$ grows. This result holds when individual idiosyncratic preference shocks are uncorrelated across both individuals and time periods, and supply shocks are uncorrelated with individual idiosyncratic preferences across both individuals and time periods. As in the linear model without a time effect, the inconsistency will disappear at the rate $O(M^{-1})$.

With the additional time effect, the key condition for the inconsistency to vanish as $M$ grows is now the following one. 

\begin{assumption}\label{ass2}
\begin{enumerate}[(i)]
    \item $Cov(\varepsilon_{it}, \eta_{js}) = 0$ for all $i,j,t,s$ in the sample;
    \item $E[\eta_{it}] = 0$, including for individuals not in the sample;
    \item $\eta_{it}$ is i.i.d. over $i$ and $t$, including individuals $i$ not in the sample, and $\varepsilon_{it}$ is i.i.d. over $i$ and $t$ for individuals $i$ in the sample.
\end{enumerate}
\end{assumption}

Assumption \ref{ass2} states that the market supply shock $\varepsilon_{it}$ is orthogonal to the individual preference shock $\eta_{js}$ and that preference shocks are uncorrelated across time and individuals. Assumption \ref{ass2} also requires that both $\eta_{it}$ and $\varepsilon_{it}$ are identically distributed over $i$ and $t$, although this assumption is for notational simplicity and can be relaxed, as discussed later.

Under this condition, we obtain the following result:
\begin{theorem} \label{thm:linear_time}
Suppose Assumption \ref{ass2} holds. Then the slope estimate $\hat\beta$ from the regression of $\tilde q_{it}$ on $\tilde p_{it}$ has probability limit as $N,T \to \infty$
\begin{align*}
    plim(\hat\beta) &= \beta + \frac{1}{M}\frac{-1}{\beta-\delta}\frac{Var(\eta_{it})}{E[\tilde p_{it}^2]}.
\end{align*}
\end{theorem}
Note that because we sample individuals from independent markets, there is sufficient independent variation in $p_{it}$ across observations to make the usual formula for the probability limit of OLS valid.
From this formula we see that as the number of individuals $M$ in each market grows $plim(\hat\beta)$ converges to the true $\beta$ at rate $O(M^{-1})$. 

This result shows that even with a linear demand model that contains a fixed time effect, under certain conditions, the price endogeneity caused by simultaneity disappears as the number of individuals in all markets grows large. The requirements for this result to hold are stricter with the additional time effect. In particular, we now require orthogonality of the supply shock and preference shock to hold over time. Importantly, just as before, the individual preference component $\alpha_i$ of the unobserved determinant of preferences $\alpha_i + \phi_t + \eta_{it}$ is allowed to covary with all other variables, like the supply shock or the idiosyncratic preference shock. Note that Assumption \ref{ass2}(iii), which implies $Var(\eta_{it})$ is the same across all $i$ and $t$, is not crucial for the $O(M^{-1})$ rate to hold, as long as $Var(\eta_{it})$ can be uniformly upper bounded by a constant.

In order to have sufficient independent price variation we made the sampling assumption that each individual in the sample came from a different market. This sampling assumption can be relaxed to allow the sample to contain multiple individuals from the same market, so long as there remains sufficient independent price variation across individuals in the sample. We can also relax the assumption that all markets have exactly $M$ individuals; the result holds taking $M$ to be the size of the smallest market.

\section{A Random Coefficient Demand Model for Panel Data}
\label{sec4}

In this section, we extend the linear demand model to allow for a random coefficient on price. We show under similar conditions as before that the bias of the OLS slope estimator for the mean of the random coefficient will disappear as the number of consumers in the market grows.

As before, consider a single market with $M$ consumers $i=1,\dots,M$. For a sample of size $N$ of those consumers, we observe quantity $q_{it}$ and price $p_t$ of a single good for each consumer $i$ over $T \geq 2$ time periods $t = 1, \dots, T$.

Suppose we can write demand for every consumer in the market as
\begin{align*}
    q_{it} &= \beta_ip_t + \eta_{it},
\end{align*}
where $\beta_i$ is a random coefficient on price. 
Average demand in period $t$, where the average is taken over all individuals in the market, is
\begin{align*}
    \bar q_{t} = \bar\beta p_t + \bar\eta_t, \qquad \text{for } \bar q_{t} \equiv \frac{1}{M} \sum_{i=1}^M q_{it}, \bar \beta \equiv \frac{1}{M} \sum_{i=1}^M \beta_i, \bar\eta_t \equiv \frac{1}{M} \sum_{i=1}^M \eta_{it}.
\end{align*}

Suppose as before that we can write the inverse supply model as a linear model of average supply,
\begin{align*}
    \bar q_{t} &= \gamma + \delta p_{t} + \varepsilon_{t}.
\end{align*}
Solving for equilibrium price $p_t$ by setting demand equal to supply then gives
\begin{align*}
    p_t &= \frac{1}{\bar\beta - \delta} (\gamma + \varepsilon_{t} - \bar\eta_t ).
\end{align*}

We will maintain the following key assumptions.
\begin{assumption}\label{ass3}
    \begin{enumerate}[(i)]
    \item $\beta_i = \beta_0 + \omega_i$, for $\omega_i$ an i.i.d. random variable with $E[\omega_i] = 0$;
    \item $\eta_{it}$ is mean zero and i.i.d. over $i$ and $t$ and $\varepsilon_{t}$ is independent across $t$, with $Cov(\eta_{it}, \varepsilon_t) = 0$ for all $i, t$;
    \item $\{\omega_i\}_{i=1}^M$ is independent of $\{(\varepsilon_t, \eta_{1t}, \dots, \eta_{Mt})\}_{t=1}^T$.
\end{enumerate}
\end{assumption}

Assumption \ref{ass3} maintains similar assumptions to those imposed previously on preference and supply shocks, but additionally requires that the random coefficient be independent of both preference and supply shocks.

Our goal will be to estimate $\beta_0$, the mean of the random coefficient. We will show that the slope coefficient from the OLS regression of $q_{it}$ on $p_t$ has inconsistency that disappears with the number of consumers in the market $M$.

\begin{theorem}\label{thm:random}
    Under Assumption \ref{ass3} and an additional technical assumption on $\beta_i$ detailed in the Appendix, the slope estimate $\hat\beta$ from the regression of $q_{it}$ on $p_t$ has probability limit as $N,T \to \infty$ 
    \begin{align*}
        plim(\hat\beta) &= \beta_0 + \frac{1}{M} E \left[ \frac{-1}{\bar\beta - \delta} \right] \frac{E[\eta_{it}^2]}{E[p_t^2]}.
    \end{align*}
\end{theorem}
The additional technical assumption required for the result ensures that moments of equilibrium prices exist and are uniformly bounded.
Under this additional assumption, we obtain a rate of convergence of $O(1/M)$, just as for the linear demand model.

\section{A Random Coefficient Demand Model with Time Effect} 
\label{sec5}

In this section, we extend the random coefficient demand model from the previous section to allow for a fixed time effect. We show under similar conditions as before that the bias of the OLS slope estimator for the mean of the random coefficient will disappear as the number of consumers in each market grows. 

As before, suppose there is a good sold in many markets, each of which has $M$ consumers. Consider a sample of $N$ individuals, $i = 1,\dots,N$, each from a different market $m$ and observed over $T \geq 2$ time periods $t = 1,\dots, T$. In the sample we observe the quantity demanded $q_{it}$ and price $p_{it}$ for each individual $i$ in each time period $t$. As before, note that while we assume prices do not vary across individuals in the same market, in the sample prices vary across individuals who come from different markets.

Suppose that for each individual in every market, demand can be written as
\begin{align*}
    q_{it} &= \beta_i p_{it} + \phi_t + \eta_{it},
\end{align*}
where $\beta_i$ is a random coefficient and $\phi_t$ is a fixed time effect. 
Average demand in period $t$ and market $m$, where the average is taken over the individuals in market $m$ and $p_{mt}$ denotes the price in market $m$ and period $t$, is
\begin{align*}
    \bar q_{mt} = \bar\beta_m p_{mt} + \phi_t + \bar\eta_{mt}, \qquad \text{for } \bar q_{mt} \equiv \frac{1}{M} \sum_{i \in m} q_{it}, \bar \beta_m \equiv \frac{1}{M} \sum_{i \in m} \beta_i, \bar\eta_{mt} \equiv \frac{1}{M} \sum_{i \in m} \eta_{it}.
\end{align*}

Suppose as before that in each market $m$ we can write the inverse supply model as a linear model of average supply,
\begin{align*}
    \bar q_{mt} &= \gamma + \delta p_{mt} + \varepsilon_{mt},
\end{align*}
where $\varepsilon_{mt}$ is a supply shock term. Solving for equilibrium prices $p_{mt}$ in market $m$ by setting demand equal to supply, we obtain
\begin{align*}
    p_{mt} = \frac{1}{\bar\beta_m-\delta} (\gamma + \varepsilon_{mt} - \phi_t - \bar{\eta}_{mt} ).
\end{align*}

We will maintain the following key assumptions.
\begin{assumption}\label{ass4}
    \begin{enumerate}[(i)]
    \item $\beta_i = \beta_0 + \omega_i$, for $\omega_i$ an i.i.d. random variable with $E[\omega_i] = 0$;
    \item $\eta_{it}$ is mean zero and i.i.d. across $i$ in all markets and $t$ and $\varepsilon_{mt}$ is independent across $m$ and $t$, with $Cov(\eta_{it}, \varepsilon_{ms}) = 0$ for all $i, t, s, m$;
    \item $\{\omega_i: i \text{ in any market}\}$ is independent of $(\{\eta_{jt}: j \text{ in any market}, t = 1, \dots, T\}, \{\varepsilon_{mt}: m \text{ is any market}, t = 1,\dots,T\})$.
\end{enumerate}
\end{assumption}
We let $\{\eta_{jt}: j \text{ in any market}, t = 1, \dots, T\}$ denote the collection of $\eta_{it}$ across all periods $t$ and all individuals $j$ in every single market, and we let $\{\varepsilon_{mt}: m \text{ is any market}, t = 1,\dots,T\}$ denote the collection of $\varepsilon_{mt}$ across all periods $t$ and all markets $m$.
Assumption \ref{ass4} maintains similar assumptions to Assumption \ref{ass3} but requires stronger orthogonality of preference and supply shocks across all individuals in all markets and across all time periods. 

As before, our goal will be to estimate $\beta_0$, the mean of the random coefficient. We will consider the slope estimate from the regression of $q_{it}$ on $p_{it}$ and time dummies.

\begin{theorem}\label{thm:random_time}
    Under Assumption \ref{ass4} and additional technical assumptions detailed in the Appendix, the slope estimate $\hat\beta$ from the regression of $q_{it}$ on $p_{it}$ and time dummies has probability limit as $N,T \to \infty$
    \begin{align*}
        plim(\hat\beta) &= \beta_0 + O(1/M).
    \end{align*}
\end{theorem}
The additional technical assumptions ensure that sufficient moments of prices and demeaned prices exist and are uniformly bounded.

\section{Conclusion}

In this paper we showed how panel data can be used to estimate individual demand models without instrumental variables, despite prices being simultaneously determined in markets. 
This result relies on the key assumption that idiosyncratic individual demand shocks are orthogonal to supply shocks and other individuals' idiosyncratic demand shocks, which is plausible in a wide variety of settings.
We considered linear demand models and random coefficient demand models, both with and without a fixed time effect, together with a linear supply model. We characterized the rate at which bias of individual demand estimates shrinks as the number of individuals in each market grows. 

An important direction for future work is extending these results from linear demand models to more general models of demand and from linear supply models to more general models of supply. We are currently working on such extensions, which rely on the intuition that individual-specific random coefficient models only require time homogeneity of individual preferences. Showing these results  requires different proof techniques.

\small
\singlespacing
{
\bibliographystyle{chicago} 
\bibliography{ref}
}

\appendix
\section{Appendix}

\subsection*{Proof of Theorem \ref{thm:linear_time}}
\begin{proof}

Recall that 
\begin{align*}
    \underbrace{q_{it} - \bar q_t - \bar q_i + \bar q}_{\displaystyle \tilde q_{it}} &= \beta \underbrace{\left(p_{it} - \bar p_t - \bar p_i + \bar p \right)}_{\displaystyle \tilde p_{it}} + \underbrace{\eta_{it} - \bar\eta_t - \bar\eta_i + \bar\eta}_{\displaystyle \tilde \eta_{it}}, \\
    \underbrace{\left(p_{it} - \bar p_t - \bar p_i + \bar p \right)}_{\displaystyle \tilde p_{it}} &= \frac{1}{\beta-\delta} \big( \underbrace{\varepsilon_{it} - \bar \varepsilon_t - \bar \varepsilon_i + \bar \varepsilon}_{\displaystyle \tilde \varepsilon_{it}} - (\underbrace{\bar\eta_{it}^m - \bar \eta_t^m - \bar \eta_i^m + \bar \eta^m}_{\displaystyle \tilde \eta_{it}^m}) \big),
\end{align*}
defining 
\begin{align*}
    \bar q_t &= \frac{1}{N} \sum_{i=1}^N q_{it}, \quad \bar q_i = \frac{1}{T} \sum_{t=1}^T q_{it}, \quad \bar q = \frac{1}{NT} \sum_{i=1}^N \sum_{t=1}^T q_{it}, \\
    \bar p_t &= \frac{1}{N} \sum_{i=1}^N p_{it}, \quad \bar p_i = \frac{1}{T} \sum_{t=1}^T p_{it}, \quad \bar p = \frac{1}{NT} \sum_{i=1}^N \sum_{t=1}^T p_{it}, \\
    \bar \eta_t &= \frac{1}{N} \sum_{i=1}^N \eta_{it}, \quad \bar \eta_i = \frac{1}{T} \sum_{t=1}^T \eta_{it}, \quad \bar \eta = \frac{1}{NT} \sum_{i=1}^N \sum_{t=1}^T \eta_{it}, \\
    \bar \varepsilon_t &= \frac{1}{N} \sum_{i=1}^N \varepsilon_{it}, \quad \bar \varepsilon_i = \frac{1}{T} \sum_{t=1}^T \varepsilon_{it}, \quad \bar \varepsilon = \frac{1}{NT} \sum_{i=1}^N \sum_{t=1}^T \varepsilon_{it}, \\
    \bar \eta_t^m &= \frac{1}{N} \sum_{i=1}^N \bar\eta_{it}^m, \quad \bar \eta_i^m = \frac{1}{T} \sum_{t=1}^T \bar\eta_{it}^m, \quad \bar \eta^m = \frac{1}{NT} \sum_{i=1}^N \sum_{t=1}^T \bar\eta_{it}^m.
\end{align*}
For individual $i$ in market $m$, under Assumption \ref{ass2} it follows that
\begin{align*}
    E[\tilde p_{it} \tilde \eta_{it}] 
    &= \frac{1}{\beta-\delta} \left( E[\tilde \varepsilon_{it} \tilde \eta_{it}] - E[ \tilde \eta_{it}^m \tilde \eta_{it}] \right).
\end{align*}

Note that under Assumption \ref{ass2}, 
\begin{align*}
    E[ \tilde \varepsilon_{it} \tilde \eta_{it}] &= Cov(\varepsilon_{it},\eta_{it}) - Cov(\varepsilon_{it},\bar\eta_t) - Cov(\varepsilon_{it},\bar\eta_i) + Cov(\varepsilon_{it},\bar\eta) \\
    & \quad -Cov(\bar\varepsilon_t,\eta_{it}) + Cov(\bar\varepsilon_t,\bar\eta_t) + Cov(\bar\varepsilon_t,\bar\eta_i) - Cov(\bar\varepsilon_t,\bar\eta) \\
    & \quad -Cov(\bar\varepsilon_i,\eta_{it}) + Cov(\bar\varepsilon_i,\bar\eta_t) + Cov(\bar\varepsilon_i,\bar\eta_i) - Cov(\bar\varepsilon_i,\bar\eta) \\
    & \quad + Cov(\bar\varepsilon,\eta_{it}) - Cov(\bar\varepsilon,\bar\eta_t) - Cov(\bar\varepsilon,\bar\eta_i) + Cov(\bar\varepsilon,\bar\eta) \\
    &= 0.
\end{align*}

For individual $i$ in market $m$,
\begin{align*}
    Cov(\bar\eta_{it}^m, \eta_{it}) &= \frac{1}{M} \sum_{j \in m} Cov(\eta_{jt}, \eta_{it}) = \frac{1}{M} Var(\eta_{it}).
\end{align*}
Then
\begin{align*}
    E[ \tilde \eta_{it}^m \tilde \eta_{it}] &= Cov(\bar\eta^m_{it},\eta_{it}) - Cov(\bar\eta^m_{it},\bar\eta_t) - Cov(\bar\eta^m_{it},\bar\eta_i) + Cov(\bar\eta^m_{it},\bar\eta) \\
    & \quad -Cov(\bar\eta^m_t,\eta_{it}) + Cov(\bar\eta^m_t,\bar\eta_t) + Cov(\bar\eta^m_t,\bar\eta_i) - Cov(\bar\eta^m_t,\bar\eta) \\
    & \quad -Cov(\bar\eta^m_i,\eta_{it}) + Cov(\bar\eta^m_i,\bar\eta_t) + Cov(\bar\eta^m_i,\bar\eta_i) - Cov(\bar\eta^m_i,\bar\eta) \\
    & \quad + Cov(\bar\eta^m,\eta_{it}) - Cov(\bar\eta^m,\bar\eta_t) - Cov(\bar\eta^m,\bar\eta_i) + Cov(\bar\eta^m,\bar\eta) \\
    &= Cov(\bar\eta^m_{it},\eta_{it}) - \frac{1}{N} Cov(\bar\eta^m_{it},\eta_{it}) - \frac{1}{T} Cov(\bar\eta^m_{it},\eta_{it}) + \frac{1}{NT} Cov(\bar\eta^m_{it},\eta_{it}) \\
    & \quad -\frac{1}{N} Cov(\bar\eta^m_{it},\eta_{it}) + \frac{1}{N} Cov(\bar\eta^m_{it},\eta_{it}) + \frac{1}{NT} Cov(\bar\eta^m_{it},\eta_{it}) - \frac{1}{NT} Cov(\bar\eta^m_{it},\eta_{it}) \\
    & \quad -\frac{1}{T} Cov(\bar\eta^m_{it},\eta_{it}) + \frac{1}{NT} Cov(\bar\eta^m_{it},\eta_{it}) + \frac{1}{T} Cov(\bar\eta^m_{it},\eta_{it}) - \frac{1}{NT} Cov(\bar\eta^m_{it},\eta_{it}) \\
    & \quad + \frac{1}{NT} Cov(\bar\eta^m_{it},\eta_{it}) - \frac{1}{NT} Cov(\bar\eta^m_{it},\eta_{it}) - \frac{1}{NT} Cov(\bar\eta^m_{it},\eta_{it}) + \frac{1}{NT} Cov(\bar\eta^m_{it},\eta_{it}) \\
    &= Cov(\bar\eta^m_{it},\eta_{it}) \left(1 - \frac{1}{N} - \frac{1}{T} + \frac{1}{NT} \right) \\
    &= Var(\eta_{it}) \left(\frac{1}{M} - \frac{1}{MN} - \frac{1}{MT} + \frac{1}{MNT} \right)
\end{align*}

Thus
\begin{align*}
    E[\tilde p_{it} \tilde \eta_{it}] &= \frac{-1}{\beta-\delta}Var(\eta_{it}) \left(\frac{1}{M} - \frac{1}{MN} - \frac{1}{MT} + \frac{1}{MNT} \right)
\end{align*}
which has limit as $N,T \to \infty$, keeping $M$ fixed,
\begin{align*}
    \frac{1}{M}\frac{-1}{\beta-\delta}Var(\eta_{it}).
\end{align*}

Then the OLS estimator has probability limit as $N,T \to \infty$
\begin{align*}
    plim(\hat\beta)=\beta+\frac{E[\tilde p_{it}\tilde \eta_{it}] }{E[\tilde p_{it}^2]}
    &=\beta + \frac{1}{M}\frac{-1}{\beta-\delta}\frac{Var(\eta_{it})}{E[\tilde p_{it}^2]}.
\end{align*}

\end{proof}

\subsection*{Proof of Theorem \ref{thm:random}}
\begin{proof}
As an additional regularity assumption we will require that $\omega_i, \varepsilon_{t}, \eta_{it}$ have finite second moments. We will also impose that for all $M$, there exists some $\varepsilon > 0$ such that $|\bar \beta - \delta| \geq \varepsilon$ with probability 1. These assumptions ensure $\displaystyle E \left[ \frac{1}{\bar\beta - \delta} \right]$ and second moments of price $p_t$ exist and are finite. \\

The OLS coefficient is given by
\begin{align*}
    \hat\beta &= \frac{\frac{1}{T} \sum_{t=1}^T \frac{1}{N} \sum_{i=1}^N p_t q_{it}}{\frac{1}{T} \sum_{t=1}^T p_t^2} \\
    &= \frac{\left(\frac{1}{T} \sum_{t=1}^T p_t^2 \right) \left(\frac{1}{N} \sum_{i=1}^N \beta_i\right)}{\frac{1}{T} \sum_{t=1}^T p_t^2} + \frac{\frac{1}{T} \sum_{t=1}^T \frac{1}{N} \sum_{i=1}^N p_t\eta_{it}}{\frac{1}{T} \sum_{t=1}^T p_t^2}.
\end{align*}
Then as $N,T \to \infty$,
\begin{align*}
    plim(\hat\beta) &= E[\beta_i] + \frac{E[p_t \eta_{it}]}{E[p_t^2]} \\
    &= \beta_0 + \frac{E[p_t \eta_{it}]}{E[p_t^2]}.
\end{align*}

Note that
\begin{align*}
    E[p_t \eta_{it}] &= \gamma E \left[\frac{\eta_{it}}{\bar\beta - \delta} \right] + E \left[ \frac{\varepsilon_t\eta_{it}}{\bar\beta - \delta} \right] - E \left[ \frac{\bar\eta_t\eta_{it}}{\bar\beta - \delta} \right] \\
    &= \gamma E \left[\frac{1}{\bar\beta - \delta} \right]E[\eta_{it}] + E \left[ \frac{1}{\bar\beta - \delta} \right] E [\varepsilon_t\eta_{it}] - E \left[ \frac{1}{\bar\beta - \delta} \right]E[\bar\eta_t\eta_{it}] \\
    &= - \frac{1}{M} E \left[ \frac{1}{\bar\beta - \delta} \right]E[\eta_{it}^2]
\end{align*}
using independence Assumption \ref{ass3}(iii), that $Cov(\eta_{it},\varepsilon_t) = 0$, that $\eta_{it}$ is i.i.d. over $i$ and $t$, and that $\eta_{it}$ is mean zero.
\end{proof}

\subsection*{Proof of Theorem \ref{thm:random_time}}
\begin{proof}
As an additional regularity assumption we will require that $\omega_i, \varepsilon_{mt}, \eta_{it}$ all have uniformly bounded fourth moment and that $\phi_t$ is uniformly bounded. We will also impose that there exists an $\epsilon > 0$ such that for all $\rho \in [0,1]$ and for all $M$, $|\beta_0 - \delta + \frac{1}{M} \sum_{j \in m, j \neq i} \omega_j + \rho \omega_i/M| \geq \epsilon$ almost surely. 

Finally, for notational simplicity we assume that $E[\tilde p_{it}^2], E[\tilde p_{it}^2\omega_i], E[\tilde p_{it}\bar p_t \omega_i],$ and $E[\tilde p_{it}\eta_{it}]$ do not depend on $t$ and that $E[\tilde p_{it}^2]$ is bounded away from zero, where $\tilde p_{it} = p_{it} - \frac{1}{N} \sum_{i=1}^N p_{it}$. The assumption that moments do not depend on $t$ can be easily relaxed at the cost of additional notation. \\

By Frisch-Waugh-Lovell the slope coefficient from the regression of $q_{it}$ on $p_{it}$ and time dummies is
\begin{align*}
    \hat\beta &= \frac{\sum_{i=1}^N \sum_{t=1}^T \tilde p_{it} q_{it}}{\sum_{i=1}^N \sum_{t=1}^T \tilde p_{it} ^2} \\
    &= \frac{\sum_{i=1}^N \sum_{t=1}^T \tilde p_{it} (\beta_i p_{it} + \eta_{it})}{\sum_{i=1}^N \sum_{t=1}^T \tilde p_{it} ^2} \\
    &= \frac{\sum_{i=1}^N \sum_{t=1}^T (\tilde p_{it}^2 \beta_i + \tilde p_{it} \bar p_t \beta_i + \tilde p_{it} \eta_{it})}{\sum_{i=1}^N \sum_{t=1}^T \tilde p_{it}^2} \\
    &= \beta_0 + \frac{\sum_{i=1}^N \sum_{t=1}^T \tilde p_{it}^2 \omega_i}{\sum_{i=1}^N \sum_{t=1}^T \tilde p_{it}^2} + \frac{\sum_{i=1}^N \sum_{t=1}^T \tilde p_{it} \bar p_t \omega_i}{\sum_{i=1}^N \sum_{t=1}^T \tilde p_{it}^2} + \frac{\sum_{i=1}^N \sum_{t=1}^T \tilde p_{it} \eta_{it}}{\sum_{i=1}^N \sum_{t=1}^T \tilde p_{it}^2},
\end{align*}
where $\tilde p_{it} = p_{it} - \bar p_t$ and $\bar p_t \equiv \frac{1}{N} \sum_{i=1}^N p_{it}$, so $\sum_{i,t} \tilde p_{it} \bar p_t \beta_0 = 0$.

Then as $N,T \to \infty$,
\begin{align*}
    plim(\hat\beta) &= \beta_0 + \frac{E[\tilde p_{it}^2 \omega_i]}{E[\tilde p_{it}^2]} + \frac{E[\tilde p_{it} \bar p_t \omega_i]}{E[\tilde p_{it}^2]} + \frac{E[\tilde p_{it} \eta_{it}]}{E[\tilde p_{it}^2]}.
\end{align*}

Term $E[\tilde p_{it} \eta_{it}]$:

Recall the equilibrium price expression: we know
\begin{align*}
    p_{it} &= \frac{1}{\bar \beta_i - \delta} (\gamma + \varepsilon_{it} - \phi_t - \bar \eta_{it}),
\end{align*}
where we index $\bar\beta_i,\varepsilon_{it},$ and $\bar\eta_{it}$ by $i$ because each individual $i$ comes from a different market $m$.

Note 
\begin{align*}
    E[\tilde p_{it} \eta_{it}] &= E[p_{it}\eta_{it}] - E[\bar p_t \eta_{it}] \\
    &= E[p_{it}\eta_{it}] - \frac{1}{N} E[p_{it} \eta_{it}] \\
    &= -\frac{N-1}{N} E\left[\frac{\bar\eta_{it}\eta_{it}}{\bar\beta_i - \delta}\right] \\
    &= -\frac{N-1}{N} \frac{1}{M} E\left[\frac{1}{\bar\beta_i - \delta}\right] E[\eta_{it}^2]
\end{align*}
under Assumptions \ref{ass4}(ii) and \ref{ass4}(iii).
But since we've taken $N,T \to \infty$, this means
\begin{align*}
    E[\tilde p_{it} \eta_{it}] &= O(1/M)
\end{align*}
under the additional assumptions, which ensure $\displaystyle E\left[\frac{1}{\bar\beta_i - \delta}\right]$ and $E[\eta_{it}^2]$ are uniformly bounded. \\

For the other two terms, we must deal with the fact that equilibrium prices depend on $\beta_i$ and thus $\omega_i$. We would like to show that this dependence is ``small,'' that is, $O(1/M)$.

For this, we will define additional notation. Define
\begin{align*}
    \omega_{-i} &\equiv \frac{1}{M} \sum_{j \in m, j \neq i} \omega_j, \qquad u_{it} \equiv \gamma + \varepsilon_{it} - \phi_t - \bar \eta_{it}, \qquad h(x) \equiv \frac{1}{x + \beta_0 - \delta},
\end{align*}
so that $p_{it} = u_{it} h(\omega_i/M + \omega_{-i}).$
Define $$\check p_{it} = u_{it} h(\omega_{-i})$$ to be the ``$\omega_i$ leave out'' price, and let $$\bar{\check{p}}_{t} = \frac{1}{N} \sum_{i=1}^N \check p_{it}.$$ Then define $$\tilde{\check{p}}_{it} \equiv \check p_{it} - \bar{\check{p}}_{t}.$$

Note that by Assumption \ref{ass4}(i) and (iii) we have $\omega_i$ is independent of $(\omega_{-i}, u_{it})$. Thus $\omega_i$ is independent of $(\check{p}_{it}, \bar{\check{p}}_t,\tilde{\check{p}}_{it})$.

Furthermore, by the additional technical assumption we have $|h(x)| \leq 1/\epsilon$ and $|h'(x)| = 1/(x+\beta_0-\delta)^2 \leq 1/\epsilon^2$ for any $x$ of the form $x = \rho\omega_i/M + \omega_{-i}$ for some $\rho \in [0,1]$.

Thus $|p_{it}| \leq |u_{it}|/\epsilon$ and $|\check{p}_{it}| \leq |u_{it}|/\epsilon$.  \\

Term $E[\tilde p_{it}^2 \omega_i]$:
Because $\tilde {\check{p}}_{it}$ is independent of $\omega_i$, which has zero mean,
\begin{align*}
    E[\tilde p_{it}^2 \omega_i] &= E[(\tilde p_{it}^2 - \tilde{\check{p}}_{it}^2) \omega_i] + E[\tilde {\check{p}}_{it}^2 \omega_i] 
    = E[(\tilde p_{it} + \tilde{\check{p}}_{it})(\tilde p_{it} - \tilde{\check{p}}_{it}) \omega_i] \\
    \Rightarrow |E[\tilde p_{it}^2 \omega_i]| &\leq E[|\tilde p_{it} + \tilde{\check{p}}_{it}||\tilde p_{it} - \tilde{\check{p}}_{it}||\omega_i|]
    \leq (E[|\tilde p_{it} + \tilde{\check{p}}_{it}|^4])^{1/4} (E[|\tilde p_{it} - \tilde{\check{p}}_{it}|^2])^{1/2} (E[|\omega_i|^4])^{1/4}
\end{align*}
by Jensen's and H\"older's inequalities.

By the additional regularity assumptions we know $p_{it}, \check p_{it}, \omega_i$ have uniformly bounded fourth moment, meaning $\tilde p_{it}$ and $\tilde{\check{p}}_{it}$ do as well. Thus $(E[|\tilde p_{it} + \tilde{\check{p}}_{it}|^4])^{1/4}$ and $(E[|\omega_i|^4])^{1/4}$ are $O(1)$. 

I now show that $(E[|\tilde p_{it} - \tilde{\check{p}}_{it}|^2])^{1/2} = O(1/M)$, which means $E[\tilde p_{it}^2 \omega_i] = O(1/M)$.

Note
\begin{align*}
    |\tilde p_{it} - \tilde{\check{p}}_{it}| &\leq |p_{it} - \check p_{it}| + |\bar p_t - \bar{\check{p}}_{t}| \\
    \Rightarrow E[|\tilde p_{it} - \tilde{\check{p}}_{it}|^2] &\leq 2E[|p_{it} - \check p_{it}|^2] + 2E[|\bar p_t - \bar{\check{p}}_{t}|^2].
\end{align*}
We can write
\begin{align*}
    |p_{it} - \check p_{it}| &= |u_{it}||h(\omega_i/M + \omega_{-i}) - h(\omega_{-i})|.
\end{align*}
By the mean value theorem and conditions above we have, for some $\rho \in (0,1)$,
\begin{align*}
    |h(\omega_i/M + \omega_{-i}) - h(\omega_{-i})| &= |h'(\rho\omega_i/M + \omega_{-i}) \omega_i/M| \leq \frac{|\omega_i|}{M \epsilon^2}.
\end{align*}
Then 
\begin{align*}
    E[|p_{it} - \check p_{it}|^2] &\leq \frac{1}{M^2\epsilon^4} E[|u_{it}|^2|\omega_{i}|^2] \leq \frac{1}{M^2\epsilon^4} E[|u_{it}|^4]^{1/2} E[|\omega_{i}|^4]^{1/2} = O(1/M^2)
\end{align*}
and
\begin{align*}
    E[|p_{it} - \check p_{it}|] &\leq \frac{1}{M\epsilon^2} E[|u_{it}||\omega_{i}|] \leq \frac{1}{M\epsilon^2} E[|u_{it}|^2]^{1/2} E[|\omega_{i}|^2]^{1/2} = O(1/M)
\end{align*}
by Cauchy-Schwarz and finite fourth moments. 

Meanwhile we can write 
\begin{align*}
    \bar p_t - \bar{\check{p}}_{t} &= \frac{1}{N} \sum_{i=1}^N (p_{it} - \check p_{it}) \\
    \Rightarrow |\bar p_t - \bar{\check{p}}_{t}|^2 &\leq \frac{1}{N} \sum_{i=1}^N |p_{it} - \check p_{it}|^2 \\
    \Rightarrow E[|\bar p_t - \bar{\check{p}}_{t}|^2] &\leq \frac{1}{N} \sum_{i=1}^N E[|p_{it} - \check p_{it}|^2] \leq O(1/M^2).
\end{align*}
by Cauchy-Schwarz.

Thus $E[|\tilde p_{it} - \tilde{\check{p}}_{it}|^2] = O(1/M^2)$, so $(E[|\tilde p_{it} - \tilde{\check{p}}_{it}|^2])^{1/2} = O(1/M)$. \\

Term $E[\tilde p_{it} \bar p_t \omega_i]$:
Because $\omega_i$ has zero mean and is independent of $(\check{p}_{it}, \bar{\check{p}}_t,\tilde{\check{p}}_{it})$,
\begin{align*}
    E[\tilde p_{it} \bar p_t \omega_i] &= E[(\tilde p_{it} \bar p_t - \tilde{\check{p}}_{it} \bar{\check{p}}_{t}) \omega_i] + E[\tilde{\check{p}}_{it} \bar{\check{p}}_{t} \omega_i] = E[(\tilde p_{it} - \tilde{\check{p}}_{it}) \bar p_t \omega_i] + E[\tilde{\check{p}}_{it}(\bar p_t - \bar{\check{p}}_{t}) \omega_i] \\
    \Rightarrow |E[\tilde p_{it} \bar p_t \omega_i]| &\leq E[|\tilde p_{it} - \tilde{\check{p}}_{it}| |\bar p_t| |\omega_i|] + E[|\tilde{\check{p}}_{it}| |\bar p_t - \bar{\check{p}}_{t}| |\omega_i|]
\end{align*}
by triangle and Jensen's inequalities.

Applying H\"older's inequality to each term, we obtain 
\begin{align*}
    E[|\tilde p_{it} - \tilde{\check{p}}_{it}| |\bar p_t| |\omega_i|] &\leq (E[|\tilde p_{it} - \tilde{\check{p}}_{it}|^2])^{1/2} (E[ |\bar p_t|^4 ])^{1/4} (E[|\omega_i|^4])^{1/4} \\
    E[|\tilde{\check{p}}_{it}||\bar p_t - \bar{\check{p}}_{t}|| \omega_i|] &\leq (E[|\tilde{\check{p}}_{it}|^4])^{1/4} (E[|\bar p_t - \bar{\check{p}}_{t}|^2])^{1/2} (E[|\omega_i|^4])^{1/4}.
\end{align*}
By the additional regularity assumption we know $p_{it}, \check p_{it}, \omega_i$ have uniformly bounded fourth moment, meaning $(E[| \omega_i|^4])^{1/4}, (E[ |\bar p_t|^4 ])^{1/4}, $ and $(E[|\tilde{\check{p}}_{it}|^4])^{1/4}$ are $O(1)$. 
From the above argument we know $(E[|\tilde p_{it} - \tilde{\check{p}}_{it}|^2])^{1/2}$ is $O(1/M)$. In the proof that $(E[|\tilde p_{it} - \tilde{\check{p}}_{it}|^2])^{1/2} = O(1/M)$ above we showed that $E[|\bar p_t - \bar{\check{p}}_{t}|^2] = O(1/M^2)$. 

Putting these rates together, we obtain $E[\tilde p_{it} \bar p_t \omega_i] = O(1/M)$. \\

We conclude that
\begin{align*}
    plim(\hat\beta) &= \beta_0 + \frac{E[\tilde p_{it}^2 \omega_i]}{E[\tilde p_{it}^2]} + \frac{E[\tilde p_{it} \bar p_t \omega_i]}{E[\tilde p_{it}^2]} + \frac{E[\tilde p_{it} \eta_{it}]}{E[\tilde p_{it}^2]}
    = \beta_0 + O(1/M)
\end{align*}
under the additional assumptions.
\end{proof}

\end{document}